\RequirePackage{xcolor}
\documentclass[10pt,twoside]{IEEEtran}

\usepackage{amsthm,mathtools}
\usepackage{etoolbox}

\usepackage{stix}
\usepackage{amsmath,amsfonts}

\usepackage{algorithm2e}
\RestyleAlgo{ruled}
\SetKwInput{KwData}{Input}
\SetKwInput{KwResult}{Output}
\RestyleAlgo{ruled}
\SetKwComment{Comment}{// }{//}
\DontPrintSemicolon

\usepackage{graphicx}
\usepackage{textcomp}
\usepackage[switch, mathlines]{lineno}
\usepackage{balance}
\usepackage{booktabs}
\usepackage{listings}
\usepackage{adjustbox}

\definecolor{darkgreen}{rgb}{0.0,0.5,0.0}
\definecolor{gpink}{rgb}{1,0.0,0.8}
\definecolor{barRed}{RGB}{238,123,119}
\definecolor{barBlue}{RGB}{115,115,247}
\definecolor{barGrey}{RGB}{205, 205, 205}
\definecolor{mygray}{RGB}{165, 165, 165}

\usepackage{hyperref}
\usepackage{floatflt}
\usepackage{pgfplots}
\usepackage{pgfplotstable}
\pgfplotsset{compat=1.11}
\usepackage{tikz}

\pgfplotscreateplotcyclelist{my black white}{%
solid, every mark/.append style={solid, fill=gray}, mark=*\\%
dotted, every mark/.append style={solid, fill=gray}, mark=square*\\%
densely dotted, every mark/.append style={solid, fill=gray}, mark=otimes*\\%
loosely dotted, every mark/.append style={solid, fill=gray}, mark=triangle*\\%
dashed, every mark/.append style={solid, fill=gray},mark=diamond*\\%
loosely dashed, every mark/.append style={solid, fill=gray},mark=*\\%
densely dashed, every mark/.append style={solid, fill=gray},mark=square*\\%
dashdotted, every mark/.append style={solid, fill=gray},mark=otimes*\\%
dashdotdotted, every mark/.append style={solid},mark=star\\%
densely dashdotted,every mark/.append style={solid, fill=gray},mark=diamond*\\%
}

\usetikzlibrary{positioning,fit,arrows.meta,arrows,backgrounds}
\usetikzlibrary{decorations.pathreplacing,calligraphy}
\usetikzlibrary{patterns}
\usepgfplotslibrary{fillbetween}

\usepackage[style=ieee]{biblatex}
\usepackage{ifthen}
\makeatletter
\newcounter{IEEE@bibentries}
\renewcommand\IEEEtriggeratref[1]{%
  \renewbibmacro{finentry}{%
    \stepcounter{IEEE@bibentries}%
    \ifthenelse{\equal{\value{IEEE@bibentries}}{#1}}
    {\finentry\@IEEEtriggercmd}
    {\finentry}%
  }%
}
\makeatother
\newcommand{\R}{\mathcal{R}}
\newcommand{\type}[1]{\textsf{type}(#1)}
\newcommand{\state}[1]{\textsf{\textsf{value}(#1)}}

\newcommand{\Nm}[1]{N^{^{-}}(#1)}
\newcommand{\cond}{\mathbf{X}}

\theoremstyle{definition}
\newtheorem{definition}{Definition}
\newtheorem{theorem}{Theorem}

\newtheorem{lemma}{Lemma}

\newtheorem{example}{Example}
\newtheorem{problem}{Problem}

\usepackage{enumitem}

\begin{document}
\title{Sensor Placement for Online Fault Diagnosis}
\author{Dhananjay Raju, Georgios Bakirtzis, and Ufuk Topcu
\thanks{The authors are with The University of Texas at Austin, Austin, TX, 78712, USA (e-mail: \{draju,~bakirtzis,~utopcu\}@utexas.edu). } 
}

\maketitle
\frenchspacing
\begin{abstract}

Fault diagnosis is the problem of determining a set of faulty system components that explain discrepancies between observed and expected behavior. Due to the intrinsic relation between observations and sensors placed on a system, sensors' fault diagnosis and placement are mutually dependent. Consequently, it is imperative to solve the fault diagnosis and sensor placement problems jointly. One approach to modeling systems for fault diagnosis uses answer set programming (ASP). We present a model-based approach to sensor placement for active diagnosis using ASP, where the secondary objective is to reduce the number of sensors used. The proposed method finds locations for system sensors with around 500 components in a few minutes. To address larger systems, we propose a notion of \emph{modularity} such that 
it is possible to treat each module as a separate system and solve the sensor placement problem for each module independently. Additionally, we provide a fixpoint algorithm for determining the modules of a system.
\end{abstract}

\section{Introduction}

\emph{Fault diagnosis} is critical for a system's robust and reliable operation~\cite{avizienis:2004}. Identifying faults and subsequent remedial action increases productivity and reduces maintenance costs in various industrial applications. The importance of fault diagnosis cannot be stressed enough, as a single faulty component may degrade system performance to unacceptable levels. Indeed, in the case of the 2009 Air France plane crash, ice crystals threw off the plane's airspeed sensors, and consequently, its autopilot disconnected~\cite{Breeden2022Oct}, thereby transitioning the system to a hazardous state causing an accident.

Active diagnosis detects faults during runtime, allowing corrective control actions to transition the system outside of hazardous states before they cause an accident. Fault diagnosis relies on three complementary aspects: fault detection, fault isolation, and fault identification. We study the first two tasks under the generic terminology of \emph{fault diagnosis}, where the objective is to detect faulty components in the presence of imperfect information intrinsic to sensors.

Approaches to the fault diagnosis problem broadly categorize as either model-based or experiential. The model-based diagnosis uses system models usually given in the direction from causes to effects for diagnosis~\cite{reiter:87, Davis1984Dec}. In contrast, experiential approaches codify the rules of thumb, statistical intuitions, and past experience of experts. However, the structure or design of the corresponding system under diagnosis is only weakly represented, if at all~\cite{Gao2015}. To ensure adaptability and reduce development costs, we examine the problem of fault diagnosis from a model-based perspective.

A \emph{diagnosis system} is a tuple $(\textsf{SYS},\textsf{COMP})$, where $\textsf{SYS}$ is a model (description) of the system, and $\textsf{COMP}$ is a set of observable components of the system~\cite{reiter:87}. Henceforth, when the context is clear, we shall refer to a diagnosis system as a system. For a half-adder (e.g., fig.~\ref{fig:adder}), $\textsf{SYS}$ is a description of the working of a NAND gate, i.e., its logic table in the form of first-order logic sentences and the link between the various gates. The set $\textsf{COMP}$ of observable components is $\{a$,$b$,$s$,$c\}$, i.e., the inputs and the outputs. We assume that there are diagnostic sensors providing values for these components.

One practical model-based approach to fault diagnosis uses the representation of \emph{labeled directed graphs}. The vertices correspond to system components, the labels correspond to the type of component, and the edges correspond to the connections between components. First-order predicate logic sentences capture \emph{rules} that model the working of the system~\cite{Rautenberg}. This system model is general enough to capture, for example, most analog and digital circuits, or structural dependency relations, among others.

In the diagnosis problem, given a \emph{system} with a set of sensors, the fault detection algorithm classifies a collection of \emph{faulty} components that explains the observed (at runtime) sensor values. The number of sensors and placement determines whether sensor diagnosis is possible, which we demonstrate through an example (section~\ref{sec:model_based}). To exploit this connection between sensors and diagnosis, we frame the complementary \emph{sensor placement problem} as follows: given a system, find a subset of components to place sensors such that it is possible to solve the fault diagnosis problem. Jointly addressing these two problems alleviates the effect of imperfect information caused due to a limitation on the number of sensors. 

\begin{figure}[t!]
    \centering
    \includegraphics[width=0.7\linewidth]{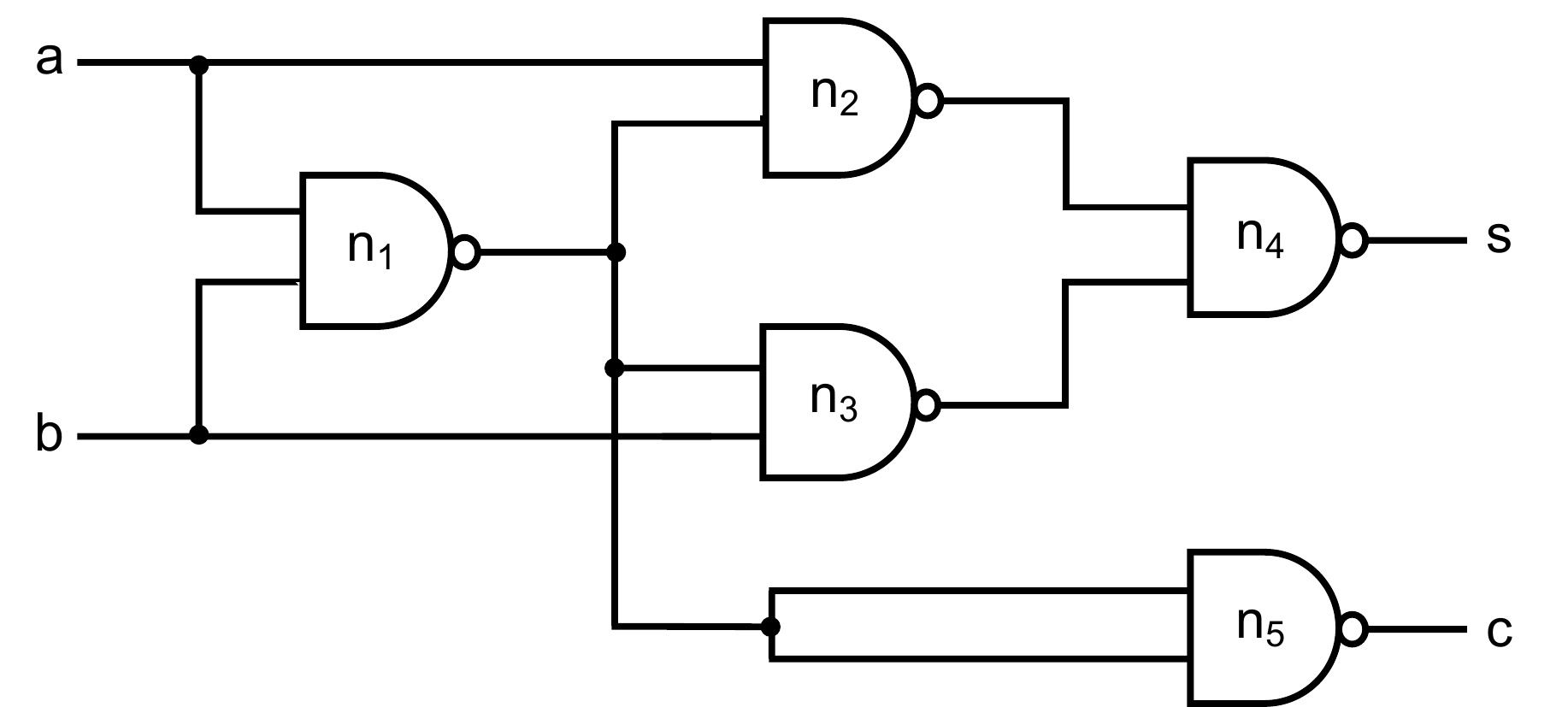}
    \caption{Half-adder circuit using NAND gates.}
    \label{fig:adder}
\end{figure}

The classic algorithm is a static procedure, i.e., it is not possible to reconfigure the system at runtime to obtain more observations. However, reconfiguring a system at runtime, for example, by changing the state of the switches in an analog circuit or inputs to a digital circuit, may produce additional observations with information that lead to greater diagnostic accuracy, i.e., better fault isolation (section~\ref{sec:model_based}). We extend the classic approach and allow runtime reconfiguration of the system subject to certain safety constraints. Intuitively, these safety constraints induce a set of configurations that are insufficient for diagnosis.

\subsubsection*{Proposed approach}

In general, a model-based approach for solving the fault diagnosis problem may not be \emph{decidable} depending on the underlying logic used to describe the rules of the system. Therefore, the key is choosing a modeling logic that is expressive enough to capture the system while also being decidable~\cite{Rautenberg}. We formulate the system model using first-order predicate logic sentences (definition~\ref{defn:rule}) that can be encoded in answer set programming (ASP), a decidable fragment of first-order predicate logic~\cite{lifschitz2003}. ASP is a logic paradigm that allows recursive definitions, weight constraints, default negation, and external atoms for solving combinatorial optimization problems~\cite{Erdem2016Oct}. The use of ASP for sensor placement has succeeded in fault diagnosis~\cite{Wotawa2022Apr}.

Model-based approaches using logical formulations follow a declarative paradigm~\cite{lifschitz2003}. To solve the sensor placement problem declaratively, we characterize the correctness of sensor locations in ASP and solve this ASP encoding. The sensors' locations are correct if, for every fault scenario (i.e., a combination of faulty components), the active diagnosis algorithm (algorithm~\ref{alg:active}) correctly detects the faulty components.

For this characterization, we formulate in ASP, 
\begin{enumerate}[label=(\alph*)]
\item fault scenarios, %
\item sensor positions, %
and
\item switch configurations. %
\end{enumerate}
We solve the sensor placement problem by solving the ASP encoding $\Pi_\text{sensor-active}$ (section~\ref{sec:sensor}) as a conjunction of these ASP encodings that can be described independently from the system to be diagnosed. %

\subsubsection*{Contributions}
Our main contributions are threefold.
\begin{enumerate}
    \item A principled approach for \emph{finding} the locations of sensors and safe configurations for active diagnosis.
    \item \emph{Reusable} ASP encodings describing the fault scenarios, switch configurations, and consistency of diagnosis. More specifically, these encodings are independent of the system and the underlying diagnosis or sensor placement problem being solved.
    \item A domain-independent notion of \emph{modularity}, which ensures that it is sufficient to solve the sensor placement problem for each \emph{module} independently of the system. Furthermore, the sensor placement problem for identical modules needs to be solved only once.
\end{enumerate}
In particular, the proposed definition of modules enables the repurposing of the classic strongly connected components algorithm for finding the modules~\cite{west_introduction_2000}. Moreover, the sensor placement algorithm (without modules) finds the positions of sensors in systems with $500$ components in a few minutes. We demonstrate the generality of the proposed approach through two scenarios: a $3$-bit adder, a digital circuit, and an electrical power system adapted from the Boeing 747 power system.

\section{Model-based fault diagnosis}
\label{sec:model_based}

In the context of model-based diagnosis, the inputs to any fault detection algorithm are a description of the system and a set of observations. In our case, the rules are expressed in first-order predicate logic (definition~\ref{defn:rule}). The output is a set of \emph{faulty} components corresponding to a minimal set of components that satisfy the following property: The assumption that each of these components is faulty, together with the assumption that all other components are behaving correctly, is consistent with the system rules and the observation.

We explain model-based diagnosis with an example (example~\ref{ex:1}). Half-adders (fig.~\ref{fig:adder}) use $s$ and $c$ to carry the result of adding two Booleans $a$ and $b$. The outputs are $s$ and $c$, where $s$ stores the sum, and $c$ stores the carry bit. Half-adders and NAND gates have straightforward output (table~\ref{tab:adder}).

\begin{table}[!t]
     \caption{The expected output for the half-adder and a NAND gate.}
    \centering
    \begin{minipage}{.45\linewidth}
    \centering
    \begin{tabular}{cc|cc}
    \toprule
         \textbf{a} & \textbf{b} & \textbf{s} & \textbf{c} \\
    \midrule
         0 & 0 & 0 & 0 \\
         0 & 1 & 1 & 0 \\
         1 & 0 & 1 & 0 \\
         1 & 1 & 0 & 1 \\
    \bottomrule 
    \end{tabular}
    \label{tab:adder}
    \end{minipage}
    \begin{minipage}{.45\linewidth}
    \centering
    \begin{tabular}{cc|c}
    \toprule
         $\textbf{in}_1$ & $\textbf{in}_2$ & \textbf{out} \\
    \midrule
         0 & 0 & 1\\
         0 & 1 & 1 \\
         1 & 0 & 1 \\
         1 & 1 & 0 \\
    \bottomrule 
    \end{tabular}
    \label{tab:nand}
    \end{minipage}
\end{table}
Say that the result of adding $a=1$ and $b=1$ using a respective adder (fig.~\ref{fig:adder}) is $s=0$ and $c=0$. We observe that the expected value of $c$ does not match the observed value of $c$. For example, this observational disparity (i.e., logical inconsistency) could be because gate $n_5$ is faulty. In the below formulation of the classic diagnosis problem, $\textsf{healthy}$ is a predicate that specifies that a component operates properly.

\begin{problem}[Fault diagnosis]
\label{defn:basic_prob}
Given a diagnosis system of the form $(\textsf{SYS}, \textsf{COMP})$ along with a set $\textsf{OBS}$ of values for the observable components in $\textsf{COMP}$, 
find a set $\Delta \subseteq V$ such that the formula
$\textsf{SYS}\cup \textsf{OBS} \cup \{ \lnot \textsf{healthy}(X)\mid X \in \Delta \} \cup \{\textsf{healthy}(X)\mid X \in V \setminus \Delta\}$
is \emph{consistent}. 
\end{problem}

Intuitively, a diagnosis (solution to the above problem) is a conjecture that certain components are faulty and the rest healthy. The problem is to specify which components we conjecture to be faulty to make the system model and observation consistent. A fundamental area for improvement of this procedure is that there may be multiple ways to restore the logical consistency of the system model and observations. 

\subsubsection*{Active diagnosis}

Let the values of the inputs and outputs remain unchanged, i.e., $a=1$, $b=1$, $s=0$ and $c=0$.
The disparity between the expected value and the observed value of $c$ can be explained under one of the following fault scenarios 1) $n_5$ is faulty, 2) $n_1$ and $n_3$ are faulty, or 3) $n_1$ and $n_5$ are faulty.
This ambiguity can be resolved when there are more observations. 
Say the output is $s=1$ and $c=1$ for the input $a=1$ and $b=0$. With this new observation, we can conclude that the fault is actually with the gate $n_5$. 
In essence, the active diagnosis procedure uses different inputs to isolate the faulty components. %

\subsubsection*{Sensor placement}
Instead of using additional inputs to generate new observations, auxiliary sensors can be used to generate the observations for fault isolation. For example, for the half-adder, if there is a sensor at $n_1$, the faulty components can be localized without using different inputs. Indeed, for the same input-output combination of $a=1$, $b=1$ and $s=0$ and $c=0$, we conclude that $n_5$ is the faulty component.

The subsequent sections formalize this intuition for solving the sensor placement problem.

\section{Modeling systems and their behaviors}
\label{sec:modeling}

In the interest of generality, we define a domain-independent concept of a system and its behavior. Formally, we study systems where a labeled digraph models the components and the connections between these components. Sentences represent the rules governing the system's behavior in the first-order predicate logic with equality~\cite{Rautenberg}. 

\begin{definition}[System]
\label{defn:SYS}
A \emph{system} is a tuple $\langle \mathcal{G},\chi \rangle$, where
\begin{itemize}
    \item $\mathcal{G} = (V,E)$ is a directed graph,
    \item $\chi\colon V \to T$ is a function that returns a component's \emph{type}. %
\end{itemize}
\end{definition}

In this system model, vertices correspond to the various system components, and the edges correspond to the connections or interfaces between components. Each component of the system belongs to a certain type. For example, in the half-adder (fig.~\ref{fig:half}), $n_1,\cdots,n_5$ are NAND gates; $a$ and $b$ are inputs; $s$ and $c$ are outputs.

Each type of component $t \in T$ has a finite set $S_t$ of states associated with it. Even though the half-adder has only one component type, other digital circuits generally have more types, like XOR gates and multiplexers, to name a couple of different types. We define the state of a component as the value of its output.
The function $\textsf{value} \colon V \to \cup_{t \in T}{S_t}$ returns the \emph{current state} of a component. %

A \emph{configuration} is any input to the system. In analog systems with switches (section~\ref{sec:case_study}), the inputs correspond to the current states of these switches, i.e., a configuration is a function that returns the switches' states (either on or off). %

A component's \emph{condition}, $v \in V$, denoted by~$\cond_v$, is a propositional formula in the conjunctive normal form on the component's type and current state. More specifically, given a subset $A \subseteq T$ of types and a subset $B \subseteq \cup_{t \in T}{S_t}$ of states, the condition~$\cond_v$ is a logical formula insisting that the type $\chi(v)$ of the component $v$ belongs to $A$ and the current state $\state{v}$ of the component belongs to $B$, i.e., 
$$
\cond_v = \bigvee_{t \in A} (\chi(v) = t) \land \bigvee_{s \in B} (\state{v} = s).
$$

\begin{example}\label{ex:1}
The following condition states that the NAND gates in the system output $1$ and XOR gates output 0.
\begin{align*}
((\type{v}  = \text{nand}) &\land (\state{v} = 1)) \\ &\lor ((\type{v}  = \text{xor}) \\ &\land (\state{v} = 0)).
\end{align*}
\end{example}

\begin{definition}[Rule]
\label{defn:rule}
A \emph{rule} is a logical formula of the form 
\begin{align*}
\left( (\chi(v) = t) \land \bigwedge_{u \in N_1}  \cond_u  \land \bigvee_{u \in N_2}  \cond_u \right)  \rightarrow 
 \bigvee_{s \in S} (\state{v} = s),
\end{align*}
where $t \in T, S \subseteq S_t, \text{ and } N_1 \uplus N_2 = \Nm{v}$. 
\end{definition}
The above rule states that if a vertex $v$ is of type $t$ such that
\begin{enumerate}[label=(\alph*)]
\item its incoming neighbors in $N_1$ must satisfy the set $\{\cond_v \colon v \in N_1 \}$ of conditions and
\item its incoming neighbors in $N_2$ must satisfy one of the conditions in $\{\cond_v : v \in N_1 \}$
then the current state of $v$ belongs to $S$.
\end{enumerate}

Associated with every system $\textsf{SYS}$ is a set $\R_{\textsf{SYS}}$ of \emph{rules} that describe the  behavior of the system. 

\subsubsection*{Behavior of the half-adder}
\label{sec:orig_adder}

To describe a system's behavior, we describe its components' behavior. Because the behaviors of any two components of the same type are identical, it is enough to describe the behavior of each type of component.

This section presents the model for the half-adder (fig.~\ref{fig:adder}). The components in the half-adder interface to other components through ports, namely the inputs and the output of the NAND gate. %
The value of the component is the value of its output port.
We make use of the following predicates to capture the values and connections between the components.\begin{itemize}
    \item $\textsf{value}(p(c),v)$ states that the value of port $p$ of component $c$ is $v$ and
    \item $\textsf{link}(p1,p2)$ states that there exists a connection between two ports $p_1$ and $p_2$.
\end{itemize}
Predicates encode the facts and inference rules for diagnosis and sensor placement problems (table~\ref{tab:predicates}).  %

\begin{table}[t!]
    \centering
    \caption{Predicates and their interpretation.}
    \begin{tabular}{ll}
        \toprule
        \textbf{Predicate} & \textbf{Meaning}  \\
        \midrule 
         $\textsf{component}(x)$ & $x$ is a component \\
         $\textsf{type}(x,t)$ & $t$ is the type of component $x$ \\
         $\textsf{link}(x,y)$ & components $x$ and $y$ are linked \\
         $\textsf{in}_1(x)$ & first input port of component $x$ \\
         $\textsf{in}_2(x)$ & second input port of component $y$ \\
         $\textsf{out}(x)$ & output port of component $x$ \\
         $\textsf{value}(x,v)$ & value of component $x$ is $v$ \\
         $\textsf{healthy}(x)$ & component $x$ is healthy, i.e., not faulty \\
         $\textsf{faulty}(x)$ & component $x$ is faulty \\
         $\textsf{switch}(x,y)$ & there is a switch on link $(x,y)$ \\
         $\textsf{on}(x,y)$ & the switch $(x,y)$ is on \\
         $\textsf{on}(i,x,y)$ & switch $(x,y)$ is on in configuration $i$ \\
         $\textsf{value}(x,v)$ & expected value of component $x$ is $v$ \\
         $\textsf{inferred}(x,v)$ & inferred value of component $x$ is $v$ \\
         $\textsf{value}(i,x,v)$ & expected value of $x$ in configuration $i$ is $v$ \\
         $\textsf{inferred}(i,x,v)$ & inferred value of $x$ in configuration $i$ is $v$ \\
         $\textsf{input}(x,v)$ & observed value of component $x$ is $v$ \\
         $\textsf{input\_healthy}(x)$ & component $x$ is healthy as per the observation \\
         $\textsf{input\_faulty}(x)$ & component $x$ is  component $x$ is faulty \\ & as per the observation \\
         $\textsf{sensor}(x)$ & there is a senor at component $x$ \\
         $\textsf{fault}(s,x)$ & component $x$ is faulty in fault scenario $s$ \\
         $\textsf{value}(s,i,x,v)$ & in the fault scenario $s$, the expected value of \\&  component $x$   due to configuration $i$ is $v$ \\
         $\textsf{inferred}(s,i,x,v)$ & in the fault scenario $s$, the inferred value of \\&  component $x$   due to configuration $i$ is $v$ \\
         \bottomrule
    \end{tabular}
    \label{tab:predicates}
\end{table}

To model the behavior of the adder, we describe the behavior of a NAND gate by inspecting its truth table (table~\ref{tab:nand}).
For each row in the truth table of the NAND gate, we create a rule that corresponds to it. The following list provides the rules for the truth table.
\smallskip
\begin{enumerate}[label={(R\arabic*)},itemsep=.8em, leftmargin=*]
\item $\begin{aligned} %
    \forall X \colon \quad \textsf{component}(X) &\land \textsf{type}(X,\text{nand}) \land\\
      \textsf{value}(\textsf{in}_1(X),0) &\land \textsf{value}(\textsf{in}_2(X),0) \\
     \longrightarrow & \textsf{value}(\textsf{out}(X),1).
\end{aligned}$
\item $\begin{aligned}\forall X \colon \quad \textsf{component}(X) &\land \textsf{type}(X,\text{nand}) \land \\
     \textsf{value}(\textsf{in}_1(X),0)  &\land \textsf{value}(\textsf{in}_2(X),1)\\
     \longrightarrow & \textsf{value}(\textsf{out}(X),1).\end{aligned}$
\item $\begin{aligned}\forall X \colon \quad  \textsf{component}(X) & \land \textsf{type}(X,\text{nand}) \land \\
      \textsf{value}(\textsf{in}_1(X),1)  & \land \textsf{value}(\textsf{in}_2(X),0) \\
      \longrightarrow & \textsf{value}(\textsf{out}(X),1).\end{aligned}$
\item  $\begin{aligned}\forall X \colon  \quad  \textsf{component}(X) & \land \textsf{type}(X,\text{nand}) \land \\
      \textsf{value}(\textsf{in}_1(X),1)  &\land \textsf{value}(\textsf{in}_2(X),1) \\
      \longrightarrow & \textsf{value}(\textsf{out}(X),0).\end{aligned}$
\end{enumerate}

\smallskip

The assumption that there are no link faults and no switches in the adder allows us to write the rule for every pair of connected ports as:
\begin{align*}
    \forall P1, P2, V \colon \textsf{link}(P1,P2) \land \textsf{value}(P1,V) \rightarrow \textsf{value}(P2,V).
\end{align*}

\subsubsection*{Structure of the half-adder}
Unlike the system's behavioral description, its structure (i.e., components and connections), is unique. Hence, all the components and the connections must be encoded separately.
We use the predicates $\textsf{component}$, $\textsf{type}$, and $\textsf{link}$ to encode the structure of the electrical power system (table~\ref{tab:predicates}). %
However, for the half-adder circuit, we also require the predicates $\textsf{in}_1, \textsf{in}_2$, and $\textsf{out}$ that differentiate the input and output ports of the NAND gate.

The system for the half-adder has the following components 
\begin{align*}
 V =   &~\{n1,\cdots,n5\} \cup\\  
    &\cup \{\textsf{in}_1(n1),\cdots,\textsf{in}_1(n5)\} \cup \{\textsf{in}_2(n1),\cdots,\textsf{in}_2(n5)\} \cup\\
    & \cup\{\textsf{out}(n1),\cdots,\textsf{out}(n5)\}.
\end{align*}
The predicate $\textsf{component}$ is described by the set 
$\{\textsf{component}(X)\mid X \in V\}$.
The connections in the half-adder are specified by the $\textsf{link}$ predicate given by the set construction
$ \left\{ \textsf{link}(X,Y)\mid (X,Y) \text{ is an edge in } E \right\}.$
The $\textsf{type}$ predicate is defined by the set 
\begin{align*}
    \{  \textsf{type}(n1,\text{nand}), & \textsf{type}(n2,\text{nand}), \textsf{type}(n3,\text{nand}), \\ & \textsf{type}(n4,\text{nand}), \textsf{type}(n5,\text{nand}) \}. 
\end{align*}

The following section provides a brief overview of ASP. This overview is required to comprehensively discuss the formulation and solution of the fault diagnosis problem.

\section{Answer set programming}
\label{sec:asp}

ASP is a logic paradigm that allows recursive definitions, weight constraints, default negation, and external atoms  for solving combinatorial optimization problems~\cite{Erdem2016Oct}. 
For space limitations, we only describe a subset of ASP relevant to this article. An ASP \emph{basic encoding} is a set of rules of the form,
\[p\leftarrow q_{1},q_{2},\ldots,q_{n}\]
where $n \geq 0$, and $p , q_1 ,\cdots ,q_n$ are propositional atoms. 
The interpretation of this rule is: $p$ is true/provable if so are $q_1, \cdots , q_n$. 
When $n=0$, the rule has an empty body. 
Such rules are called facts and are usually written as $p$. 

A model or solution of an ASP basic encoding is a set $M$ of propositional atoms that intuitively contains all and only the atoms that are true. Formally, $M$ is a solution of the basic encoding $\Pi$ if and only if $M$ is the subset-minimal set such that for every rule $p \leftarrow q_1,q_2,\cdots,q_n \in \Pi$ such that $\{q_1,\cdots,q_n\} \subseteq M$, then $p \in M$. 

A relevant syntactic element is the so-called negation as failure. Rules containing such negated atoms are of the form:
\begin{align*}p\leftarrow q_{1},q_{2},\ldots,q_{n}, \textsf{not}\: r_{1}, \textsf{not}\: r_{2}, \ldots, \textsf{not}\: r_{k}.
\end{align*}
Intuitively, we interpret the above rule as $p$ is true if $q_1, \cdots, q_n$ are true and none of $r_1, \cdots,r_k$ are true. 
The precise semantics of programs that include negation as failure is simple but a little more involved, requiring the introduction of  \emph{stable} models~\cite{lifschitz2003,FerrarisL05}. 

Another relevant type of rule is the \emph{choice} rule,
\begin{align*} \{p_{1},p_{2},\ldots,p_{n}\}=k \leftarrow q_{1},q_{2},\ldots,q_{m},\end{align*}
where $k \in \mathbb{N}^+$. 
This rule means that if $q_1,q_2,\dots,q_m$ are all true, then $k$ of the elements in $\{p_1,p_2,\cdots,p_n\}$ must appear in any solution. 
This definition allows programs to have multiple models. For example, the encoding $\{s,\{p,q,r\}=1 \leftarrow s\}$ with two rules has three solutions: $\{p,s\} ,\{q,s\}$, and $\{r,s\}$.

Another type of rule that captures impossibilities is a \emph{constraint}, given by
$\bot \leftarrow p_{1},p_{2},\ldots,p_{n}.$
This constraint prohibits the occurrence of $\{p_1,p_2,\cdots,p_n\}$ in the solution. 

ASP programs may contain variables to represent rule schemes. These schemes directly help us encode predicate logic sentences.
These rules look like
$p(X)\leftarrow q(X)$,
where uppercase letters represent \emph{variables}. 
A variable in an encoding $\Pi$ can take any value among the set of terms of $\Pi$.
As such, the above rule represents that when $c$ is a term and $q(c)$ is true, so is $p(c)$. 
Intuitively, a term represents an object (constants) that occurs in the encoding. 
The set of terms for an encoding is syntactically determined from the encoding using the objects mentioned in it (Herbrand base~\cite{Rautenberg}). 

In the process of finding a solution for an encoding, an initial step that is carried out is \emph{grounding}. Grounding instantiates rules with variables, effectively removing all variables from the encoding. %

For example, the grounded version of the ASP encoding: 
\begin{gather*}
\{q(a),q(b),p(X)\leftarrow q(X)\}  \\
\text{is } \{q(a), q(b), p(a)\leftarrow q(a), p(b) \leftarrow q(b)\}.
\end{gather*}

To find the solution for the encoding with variables, we simply find the solution for the corresponding grounded encoding.
The \emph{reduct} $F^X$ of a propositional formula $F$ with respect to $X$, a set of propositional atoms, is the formula obtained from $F$ by replacing each maximal subformula not satisfied by $X$ by $\bot$. Finally, $X$ is a \emph{stable model} of $F$ if $X$ is \emph{minimal} (set inclusion) among the sets satisfying $F^X$. $X$ is a \emph{solution} for an ASP encoding if and only if $X$ is a stable model for the ASP encoding~\cite{FerrarisL05}.

\section{The Diagnosis Problem}
\label{sec:basic}

During the operation of a system, some of its components may become faulty. However, the rules modeling the system (section~\ref{sec:modeling}) assume that the components are not faulty. 
We require another set of rules without this assumption, i.e., these rules must explicitly take into account the components' health.

Before discussing the modified rules, we present simplifying assumptions that do not affect the generality.
\begin{enumerate}
    \item Any component that is not faulty is healthy (law of excluded middle~\cite{Rautenberg}). 
\begin{align*}
     \textsf{healthy}(X) &\leftarrow \textsf{component}(X), \textsf{not faulty}(X) \text{ and }\\
     \textsf{faulty}(X) &\leftarrow  \textsf{component}(X), \textsf{not healthy}(X). 
\end{align*}
    
    \item The output of a faulty component is $0$.
    \begin{align*}
    \textsf{value}(\textsf{out}(X),0) \leftarrow \textsf{component}(X), \textsf{faulty}(X).
    \end{align*}
\end{enumerate}

The modified rules for the half-adder (NAND gate) are:

\smallskip
\begin{enumerate}[label={(R\arabic*')},itemsep=.8em, leftmargin=*]
\item $\begin{aligned}
    \textsf{value}(\textsf{out}(X),1) \leftarrow \textsf{component}(X),  \textsf{healthy}(X), \\
    \textsf{type}(X,\text{nand}), \textsf{value}(\textsf{in}_1(X),0), \textsf{value}(\textsf{in}_2(X),0).
\end{aligned}$
\item $\begin{aligned}
 \textsf{value}(\textsf{out}(X),1) \leftarrow \textsf{component}(X),  \textsf{healthy}(X), \\
    \textsf{type}(X,\text{nand}), \textsf{value}(\textsf{in}_1(X),1), \textsf{value}(\textsf{in}_2(X),0).
\end{aligned}$
\item $\begin{aligned}
 \textsf{value}(\textsf{out}(X),1) \leftarrow \textsf{component}(X),  \textsf{healthy}(X), \\
    \textsf{type}(X,\text{nand}), \textsf{value}(\textsf{in}_1(X),0), \textsf{value}(\textsf{in}_2(X),1).
\end{aligned}$
\item $\begin{aligned}
 \textsf{value}(\textsf{out}(X),0) \leftarrow \textsf{component}(X),  \textsf{healthy}(X), \\
    \textsf{type}(X,\text{nand}), \textsf{value}(\textsf{in}_1(X),1), \textsf{value}(\textsf{in}_2(X),1).
\end{aligned}$
\end{enumerate} \smallskip
The variable $\R_\text{adder}$ encapsulates the above four rules.

\subsubsection*{Fault diagnosis problem}

This section presents the formal statement of the fault diagnosis problem and provides an ASP formulation to identify the faulty components.

Recall the fault diagnosis problem formulation~(problem~\ref{defn:basic_prob}).
We refer to the set $\Delta$ of components as the \emph{diagnosis} and define \emph{minimality} of diagnosis with respect to the size of the diagnosis (i.e., $|\Delta|$).

We need a logical system expressive enough to formalize relevant knowledge and a fast solver to apply model-based reasoning in practice.
We use ASP in the input language of clingo~\cite{FerrarisL05,Gebser2014ClingoA}. The choice of encoding the problem in ASP is justified by the fact that ASP is powerful for solving default logic~\cite{ChenWZZ10} and the relation between default logic and diagnosis is well-established~\cite{Reiter1980ALF}.

In addition to the system rules $\R_{\textsf{SYS}}$ (the structure and the behavior) and the observations $\R_\text{OBS}$, we require a rule to capture the fact that every component is either faulty or healthy. We use the following ASP choice rule to capture a combination of healthy and faulty components, 
\begin{gather*}
     \{\textsf{healthy}(X)\}\leftarrow \textsf{component}(X) \text{ and }\\
     \textsf{faulty}(X) \leftarrow \textsf{component}(X), \textsf{not healthy}(X).
\end{gather*}
We denote the above two rules by $\R_\text{health}$.

The criterion of successful diagnosis is that the components' observed and expected (inferred) values match.  This requirement is captured by the following two constraints.
\begin{align*}
    & \leftarrow \textsf{value}(X,V),\textsf{not observed}(X,V) \text{ and }\\
    & \leftarrow \textsf{not value}(X,V), \textsf{observed}(X,V).
\end{align*}
We encapsulate the above rules for consistency by $\R_\text{consistency}$.

Lastly, we enforce the minimality of diagnosis with the following rule ($\R_\text{min}$).
\begin{align*}
    \#minimize\{1,X:\textsf{faulty}(X)\}.
\end{align*}

The set of faulty components correspond to the solution of the ASP encoding $\Pi_\text{basic} \cup \R_\text{OBS}$, where \[\Pi_\text{basic} = \R_{\textsf{SYS}} \cup \R_\text{health} \cup \R_\text{consistency} \cup \R_\text{min}.\]
Alternatively, a minimal diagnosis can be obtained without the above minimization rule using the incremental diagnosis algorithm (IDIAG)~\cite{Wotawa2020Sep}.

\begin{lemma}
\label{lemma:basic}
The solution of ASP encoding $\Pi_\text{basic} \cup \R_\text{OBS}$ corresponds to a diagnosis for the system $\textsf{SYS}$, i.e., if $x$ is a faulty component then $\textsf{faulty(x)}$ must belong to the solution.
\end{lemma}
\begin{proof}
Let $S$ be a solution (i.e., stable model) of the ASP encoding $\Pi_\text{basic}$ and observations $\R_{\textsf{SYS}}$. We show that if $\Delta$ is the set of faulty components, then for every component $x \notin \Delta$, $\textsf{healthy}(x) \in S$. We prove this statement by contradiction, i.e.,  say for some $x$, $x \notin \Delta$ but $\textsf{faulty}(x) \in S$.
By assumption~1, since $\textsf{faulty}(x)$ belongs to the stable model, $\textsf{healthy}(x)$ cannot belong to the stable model.
Consequently, for any rule, $\R$ that contains component $x$ in the body, the atom in the head of the rule (say $\textsf{value}(y,c)$, where $c \neq 0$) cannot be present in the stable model. 
If the observed value of component $y$ (in the encoding $\R_\text{OBS}$) is $\textsf{observed}(y,c)$, then the constraint $\leftarrow \textsf{not value}(y,c), \textsf{observed}(y,c)$ is not satisfied. Therefore, the assumption that $S$ is a stable model is violated. 
\end{proof}

\subsubsection*{Active diagnosis}
\label{sec:active}

In active diagnosis, we use multiple inputs/configurations to isolate the faulty components.
We explained the advantage of using multiple configurations previously in section~\ref{sec:model_based}.
Each configuration generates a new set of observations. Algorithm~\ref{alg:active} iteratively
solves the basic diagnosis problem for each observation (configuration). %
The final diagnosis is the \emph{hitting set} of the individual diagnoses. This procedure is equivalent to the notion of diagnosis using multiple observations~\cite{IgnatievMWM19}.

\smallskip
\begin{small}
\begin{algorithm}[t!]
\label{alg:active}
\caption{Active diagnosis}
\label{alg:two}
\KwData{$\textsf{SYS}$, set of configurations $C_1,\cdots,C_k$}
\KwResult{Set $\Delta$ of faulty components}
\For{$i \in 1 \cdots k$}{ 
$\textsf{OBS}_i \gets$ expected-value($\textsf{COMP}$)\;
$\Delta_i \gets$ diagnosis$(\textsf{SYS},\textsf{COMP},\textsf{OBS}_i)$\;
}
$\Delta \gets \bigcap_{1 \leq i \leq k} \Delta_i$\;
\end{algorithm}
\end{small}

The hardness lies in determining the configurations for active diagnosis (algorithm~\ref{alg:active}). 
We present the active diagnosis problem and solution under the assumption that a \emph{single component may fail at any given time}. 
However, this assumption does not affect the generality of the proposed approach. 

\vspace{0.2em}
\begin{definition}[Active diagnosis]
\label{defn:active}
    Given a diagnosis system $(\textsf{SYS},\textsf{COMP})$, find a set $\{C_1,C_2,\cdots,C_k\}$ of configurations s.t. if a component $v$ is faulty, then $v$ is identified as the faulty component due to algorithm~\ref{alg:active}.
\end{definition}
\vspace{0.1em}

To find the configurations for active diagnosis, 
we first enumerate all fault scenarios. 
Due to the assumption that a single component may be faulty, a system with $n \in \mathbb{N}$ components has exactly~$n$~fault scenarios. 
We generate the different fault scenarios using the following ASP rules (denoted by $\R_\text{fault}$),
\begin{gather*}
     \textsf{scenario}(S) \leftarrow 1 \leq I \leq n, \\
     \{\textsf{fault}(S,X) \colon \textsf{component}(X)\} \leftarrow  \textsf{scenario}(S) \text{ and} \\ 
     X \neq Y \leftarrow \textsf{fault}(I1,X), \textsf{fault}(I2,Y), I1 \neq I2. 
\end{gather*}
 Predicates of the above form have a precise interpretation, including $\textsf{scenario}$ and $\textsf{fault}$ (table~\ref{tab:predicates}). %

We use the following predicates for representing the configuration of switches, the expected and observed values of components. 
\begin{itemize}
    \item $\textsf{on}(i,x,y)$ states that in configuration $i$, the switch in the link between components $x$ and $y$ is on.
    \item $\textsf{observed}(s,i,x,v)$ states that in fault scenario $s$ under configuration $i$, the observed value of component $x$ is $v$.
    \item $\textsf{value}(s,i,x,v)$ states that in fault scenario $s$ under configuration $i$, the expected value of component $x$ is~$v$.
\end{itemize}

The following rules denoted by $\R_\text{config}$ characterize configuration~$I$ by the set of switches that are turned on.
\begin{gather*}
    \textsf{config}(I) \leftarrow 1 \leq I \leq k \text{ and } \\
    \{\textsf{on}(I,X,Y) \colon \textsf{switch}(X,Y)\} \leftarrow  \textsf{config}(I).
\end{gather*}

After enumerating the fault scenarios and configurations, we compute the expected value of the components for each fault scenario  using the system rules~$\R_{\textsf{SYS}}$.
The criterion for successful diagnosis remains unchanged, i.e.,  the observed and the expected values match for every fault scenario,
\begin{align*}
     &\bot \leftarrow \textsf{value}(S,X,V), \textsf{not observed}(S,X,V) \text{ and }\\
    &\bot \leftarrow \textsf{not value}(S,X,V), \textsf{observed}(X,V).
\end{align*}

\subsubsection*{Safety during active diagnosis}
\label{sec:safety}
In active diagnosis, we reconfigure the system to obtain additional observations. 
However, reconfigurations should not violate the safety constraints of a system~\cite{XuTM15}.
We encode safety constraints in ASP (section~\ref{sec:asp}) characterizing prohibited configurations.

For example, a common requirement is that AC buses should not be simultaneously powered by two AC sources. We capture this requirement by the ASP constraint
\[X1 = X2 \leftarrow \textsf{type}(Y,\text{ac\_bus}), \textsf{on}(I,X1,Y), \textsf{on}(I,X2,Y).\]

Another common safety requirement is that essential buses are always powered. It is captured by the ASP constraint,
\[  \bot  \leftarrow \textsf{component}(X), \textsf{type}(X,\text{essential}), \textsf{not powered}(X).\]

We require that any switch adjoining a generator or rectifier unit be open when that component becomes faulty. 
Using the following rule, we encode this requirement as its contrapositive (i.e., a switch is used only if the corresponding components are healthy).
\begin{align*}
    \forall X,Y, I:&~\Big{(} \textsf{switch}(X,Y) \land  \textsf{config}(I)  \land~\textsf{on}(I,X,Y) \land \\ & \land  (\textsf{type}(X, \text{source}) \lor \textsf{type}(X, \text{rectifier})) \rightarrow \\ &\rightarrow \exists J:~ \textsf{config}(J) \land J < I \land \textsf{healthy}(J,Y) \Big{)}. 
\end{align*}

The good news is that the addition of safety rules does not make the problem computationally harder~\cite{FriedrichGN90}. On the contrary, these safety constraints rule out a significant number of configurations and speed up the computation. %

We generate the set of safe configurations by solving the ASP encoding $\Pi_\text{active}$, where \[\Pi_\text{active} = \Pi_\text{basic} \cup \R_\text{fault} \cup \R_\text{config} \cup \R_\text{safety}.\]

\begin{lemma}
\label{lemma:active}
Let $\R_{\textsf{OBS}}^k,\cdots,\R_{\textsf{OBS}}^k$ be observations obtained from configurations $C_1,\cdots,C_k$, corresponding to a solution for the encoding $\Pi_\text{active}$. If component $x$ is faulty, it belongs to $\Delta$ (algorithm~\ref{alg:active}), terminating the active diagnosis.
\end{lemma}
\begin{proof}
Suppose $x$ is faulty but does not belong to $\Delta$, then there exists some $1 \leq i \leq k$, such that $\textsf{faulty}(x)$ does not belong to the solution (i.e., stable model) of $\R_{\textsf{OBS}}^i \cup \Pi_{\textsf{basic}}$. Consequently, $\textsf{faulty}(x)$ does not belong to the stable model corresponding to the configuration $C_i$. However, Lemma~\ref{lemma:basic} implies the opposite. 
\end{proof}

\section{Sensor Placement for Diagnosis}
\label{sec:sensor}

The intrinsic dependence of diagnosis on observations implies that diagnosis is also dependent on the sensors that provide these observations.
Classically, the main focus of sensor placement is the optimal placement of sensors to ensure structural observability. 
In contrast, the focus here is the placement of sensors to 
ensure the success of the fault diagnosis algorithm.
The following section discusses the sensor placement problem for the two fault diagnosis algorithms.

\begin{definition}[Sensor placement]
\label{defn:sensor_placement}
    Given a system $\textsf{SYS}$, find a subset $\textsf{COMP}$ of components such that 
    \begin{enumerate}
        \item  $|\textsf{COMP}| = m$ and
        \item  for any subset $V_{f}$ of components, if $V_f$ is faulty and $\textsf{OBS}$ is an observation, then we require $V_f \subseteq \Delta$, where $\Delta$ is the diagnosis for $(\textsf{SYS},\textsf{COMP},\textsf{OBS})$.
    \end{enumerate}
\end{definition}
\vspace{0.1em}

\subsubsection*{Sensor placement for basic diagnosis}

Although we do not know the \emph{fault scenario} (i.e., the set of faulty components) a priori, we must place sensors to identify the faulty components in every fault scenario.
 
The first step for solving the sensor placement problem using ASP is the enumeration of all the fault scenarios.
We assume that only one component may fail at a given time. %
As the rules for generating these fault scenarios do not change, we do not present them here (see $\R_\text{fault}$ in section~\ref{sec:active}).

The next step is the characterization of the sensor locations.
We directly encode the locations of sensors in ASP using the  choice rule (denoted by $\R_{senor}$)
$\{\textsf{sensor}(X) \colon \textsf{component}(X)\}=m. $
The rule simply states that $m$ different components have sensors.

Subsequently, the behavior of the components is captured by the set~$\R_{SYS}$ of rules (see section~\ref{sec:basic}). 
Finally, the criterion for the correctness of the sensor's locations does not change, i.e., we  check whether the basic diagnosis procedure can correctly identify the fault scenario. We restate these rules ($\R_\text{consistency}$).
\begin{align*}
     \bot \leftarrow& \textsf{value}(S,X,V),\textsf{not observed}(S,X,V) \text{ and }\\
     \bot \leftarrow& \textsf{not value}(S,X,V), \textsf{observed}(S,X,V).
\end{align*}
The locations of the sensors correspond to a solution for the ASP encoding $\Pi_\text{sensor-basic} = $
$\Pi_\text{basic} \cup \R_\text{fault} \cup R_\text{sensor}$.

\subsubsection*{Sensor placement for active diagnosis}

In the following section, we address the sensor placement problem for active diagnosis (algorithm~\ref{alg:active}).
\begin{figure}[t!]
\centering
\begin{tikzpicture}
\begin{axis} [ybar = .03cm,
    bar width = 3pt,
    xmin = 0,
    xmax = 500,
    ytick = {0,10,20,30,50,100,150,180,200},
    enlarge x limits = {value = .15, upper},
    enlarge y limits = {abs = .8},
    legend pos=north west,
    xlabel = number of components,
    ylabel = runtime (sec),
    width=0.9\columnwidth
]
\addplot[draw=none,fill=barBlue!80] coordinates {
(100,1.094)
(150,4.108)
(200,9.709)
(250,18.894)
(300,34.613)
(350,54.074)
(400,92.617)
(450,120.558)
(500,175.577)};
\addplot[draw=none,fill=barRed] coordinates {
(100,1.103)
(150,3.942)
(200,9.238)
(250,19.613)
(300,35.186)
(350,53.643)
(400,84.989)
(450,120.719)
(500,174.304)
};
\addplot[draw=none,fill=barGrey] coordinates {
(100,1.153)
(150,3.879)
(200,9.322)
(250,18.803)
(300,34.747)
(350,54.187)
(400,87.05)
(450,120.424)
(500,175.188)
};
\legend{$|$sensors$|$ = 10, $|$sensors$|$ = 20,$|$sensors$|$ = 30};
\end{axis}    
\end{tikzpicture}
\caption{Time to find the positions of the sensors with respect to the number of components.}
\label{fig:set1}
\end{figure}
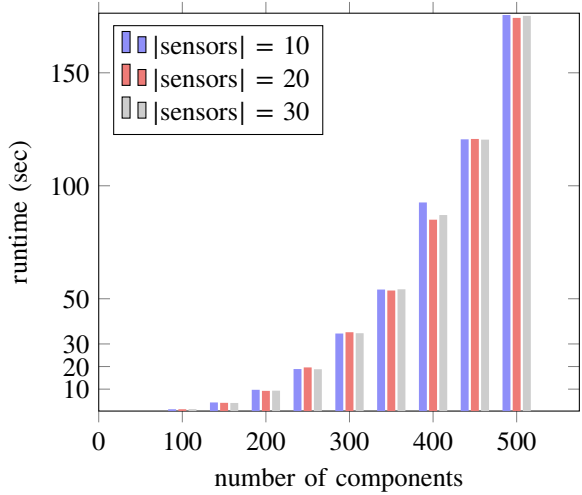

\begin{definition}[Sensor placement]
\label{defn:sensor_placement_active}
    Given a system $\textsf{SYS}$, find a subset $\textsf{COMP}$ of components such that 
    \begin{enumerate}
        \item  $|\textsf{COMP}| = m$ and
        \item  for any subset $V_{f}$ of faulty components,  then the result of active diagnosis (definition~\ref{defn:active}) classifies $V_f$ as faulty. 
    \end{enumerate}
\end{definition}

\begin{small}
\begin{algorithm}[b!]
\label{alg:sensor}
\caption{Sensor placement for active diagnosis}
\KwData{$\textsf{SYS}$, \newline $m$: max number of sensors}
\KwResult{Minimal set of positions for sensors}
\For{$i \in 1 \cdots m$}{ 
$\textsf{COMP}_i \gets$ place-sensors(\textsf{SYS},i) {\color{mygray} $\rhd$  definition~\ref{defn:sensor_placement_active}}\; 
\If{$\textsf{COMP}_i \neq \emptyset$}{
break \;
}
}
return $\textsf{COMP}_i$\;
\end{algorithm}
\end{small}

The locations of the sensors correspond to a solution for the ASP encoding $\Pi_\text{active-sensor} =$  $\Pi_\text{active} \cup \R_\text{sensor}$.

To minimize the number of sensors used, we follow an incremental approach similar to IDIAG.
In algorithm~\ref{alg:sensor}, we initially set the number of sensors to one and find a position for this sensor to make active diagnosis feasible. 
If there is no solution (i.e., the ASP solver returns unsatisfiable), then we increase the number of sensors that can be used and run the procedure again recursively until the diagnosis is possible.

\section{Experiments}
\label{sec:expts}

We demonstrate the scalability of the proposed algorithm (algorithm~\ref{alg:sensor}) for determining the positions of sensors.
The runtime of this algorithm depends on the following three parameters  a) the number of components, b) the maximum number $m$ of sensors, and c) the number of configurations.

We examine the effect of these three parameters by varying one while keeping the other two fixed.
We generate 25 random electrical power systems for each data point and record the runtime as the average over the runtimes of these 25 instances. The experiments were run on macOS with an M1 processor and 16 GB memory. %

The results of these experiments show that the number of sensors  and number of configurations have very little effect on the runtime as compared to the number of components (fig.~\ref{fig:set1}, fig.~\ref{fig:set2}a, and fig.~\ref{fig:set2}b). This observation is consistent with the combinatorial nature of the sensor placement problem because there are $|V| \choose m$ possible locations for the sensors.

The runtime of this procedure also depends on the number of possible states (section~\ref{sec:modeling}) for each of the components~\cite{Wotawa2022Apr}. 
However, we fix this number as four, negating this parameter's effect. Note that the number of different states of any component in the electrical power system is four (section~\ref{sec:modeling}). 

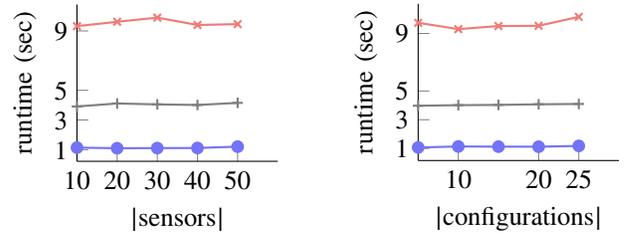
\begin{figure}[t!]
\begin{minipage}[b]{0.48\linewidth}
\centering
\begin{tikzpicture}
\begin{axis}[
     width=\columnwidth,
    axis lines*=left,
	xlabel={$|$sensors$|$},
	ylabel={runtime (sec)},
    enlarge x limits = {value = .25, upper},
    xtick={10,20,30,40,50},
    ytick = {1,3,5,9},
    legend style={at={(0.2, 0.7)},anchor=west}
    ]

\addplot[color=barBlue,mark=*,thick] coordinates {
(10,1.138)
(20,1.096)
(30,1.106)
(40,1.112)
(50,1.203)
};
\addplot[color=gray,mark=+,thick] coordinates {
(10,3.911)
(20,4.112)
(30,4.051)
(40,4.011)
(50,4.154)
};
\addplot[color=barRed,mark=x,thick] coordinates {
(10,9.323)
(20,9.618)
(30,9.897)
(40,9.402)
(50,9.459)
};
\end{axis}

\end{tikzpicture}
\hspace{0.2cm}
\end{minipage}
\begin{minipage}[b]{0.48\linewidth}
\centering
\begin{tikzpicture}
    \begin{axis}[
    width=\columnwidth,
    axis lines*=left,
	xlabel={$|$configurations$|$},
	ylabel={runtime (sec)},
    enlarge x limits = {value = .25, upper},
    xtick={10,20,25},
    ytick = {1,3,5,9},
    legend style={at={(0.2, 0.6)},anchor=west}
    ]

\addplot[color=barBlue,mark=*,thick] coordinates {
(5,1.07)
(10,1.138)
(15,1.127)
(20,1.122)
(25,1.178)
};

\addplot[color=gray,mark=+,thick] coordinates {
(5,3.973)
(10,4.013)
(15,4.03)
(20,4.074)
(25,4.096)
};
\addplot[color=barRed,mark=x,thick] coordinates {
(5,9.743)
(10,9.302)
(15,9.517)
(20,9.535)
(25,10.156)
};
\end{axis}
\end{tikzpicture}
\end{minipage}
    \caption{Time to find the positions of the sensors with respect to a) the number of sensors and b) the number of configurations for active diagnosis. The number of components is fixed as 100,150 and 200 (blue, grey and red plots).}
    \label{fig:set2}
\end{figure}

Algorithm~\ref{alg:sensor} for finding the sensor locations times out when the system has around 1000 components. However, the basic fault diagnosis algorithm works for systems with $10^4$ components~\cite{Wotawa2022Apr}. The following section discusses a method to solve the sensor placement problem for large \emph{modular} systems. %

\section{Modular structure}
\label{sec:modular}

A modular system consists of smaller parts called \emph{modules}, each of which can be examined as a separate system. The reason for considering a modular system is to solve the sensor placement problem for large systems (size ${\sim}10^4$ components). 

Digraphs represent the components and connections of a system under examination, with subgraphs capturing subsystems called modules. All \emph{mutually dependent} components lie within the same module. A component depends on another component when one's value depends on the other's value.

\begin{definition}[Module]
\label{defn:module}
Given a system $\textsf{SYS} =\langle \mathcal{G},\chi \rangle$, a system $M = \langle \mathcal{G}_M, \chi_M \rangle $ is a module of \textsf{SYS}, 
if 
\begin{itemize}
    \item $\mathcal{G}_M$ is a \emph{subgraph} of $G$,
    \item $\chi_M$ respects $\chi$, i.e., for every component $c$ in $\mathcal{G}_M$, $\chi_M(c) = \chi (c)$,
    \item if $c \in \mathcal{G}$ is a component that is strongly connected to some other component in $\mathcal{G}_M$, then $c \in \mathcal{G}_M$. 
\end{itemize} 
\end{definition}
In particular, a system $\textsf{SYS} = (\mathcal{G},\chi)$ is \emph{modular} if $G$  can be partitioned into a disjoint union of modules. A fixpoint construction returns a set of modules corresponding to $\textsf{SYS}$. We describe some useful operators first.
\begin{align*}
    M^0 & \coloneqq \{ S^j \mid S^j \text{ is a strongly connected component of } G\} \\
   V^i & \coloneqq \{ v \mid v \in S,~S \in M^i \}  \text{ and }\\
   \widetilde{V^i} & \coloneqq V \setminus V_i.
\end{align*}
Let us assume that the modules have been defined until the iteration $i$ and we inductively define the iteration ($i+1$), i.e., define $M^{i+1}$ using $M^{i}$. 
If $S_j^i \in M^i$ is the $j$-th module in $i$-th step, we define $S_j^{i+1}$ (the iterate of $S_j^i$) to be 
\begin{align*}
S_j^{i+1} = S_j^i \cup \{v \in \widetilde{V^i} \mid & v \text{ has an outgoing edge into } S_j^i  \\
& \text{ and } \text{ if } r<j, \text{  then } v \notin S_r^{i+1}\}.
\end{align*}
We terminate the iteration when a fixed point $M_i$ is reached (in at most $i \leq |V|$ steps---the number of components). 

Since each module is essentially a strongly connected subgraph, a system can be divided into modules by using the 
standard \emph{strongly connected} components algorithm from graph theory~\cite{west_introduction_2000}. If any vertex does not belong to a strongly connected subgraph, it can be added (randomly) to any of the existing modules in the previous iteration. Further, two modules can be merged to form a larger module.
The modules of a system $\textsf{SYS} = \langle \mathcal{G},\chi \rangle$ can be determined in time $\mathcal{O}(|V|(|V| + |E|))$.

\begin{small}
\begin{algorithm}[!t]
\label{alg:module}
\caption{Modular sensor placement}
\KwData{$\textsf{SYS}$, \newline $M_1,M_2,\cdots M_q$: modules \newline $m$: max number of sensors, \newline $k'$: max number of configurations}
\KwResult{$S$: Locations for sensors \\ $C_1,\cdots,C_k$: configurations for active diagnosis}
\For{$i \in 1\cdots q$}{ 
$S_i \gets$ place-sensors($M_i, k'$) {\color{mygray} $\rhd$  definition~\ref{defn:sensor_placement_active}}\; 
$S = S \cup S_i$ \; 
}
$k \gets q\cdot k'$ \;
$C_1,\cdots,C_k \gets $ $k$-configurations($\textsf{SYS},S$) {\color{mygray} $\rhd$  definition~\ref{defn:active}}\;
return $S, (C_1,\cdots,C_k)$ \;
\end{algorithm}
\end{small}

To solve the sensor placement problem for a modular system consisting of modules $M_1, M_2,\cdots, M_q$ ($q \in \mathbb{N})$, we solve the sensor placement problem for each of the modules. Subsequently, we use these sensors to find the configurations required for active diagnosis. The permitted number of configurations (for active diagnosis) $k = q \cdot k'$, where $k'$ is the maximum number of configurations required to find faulty components in any module $M_i$ ($1 \leq i \leq q)$ (algorithm~\ref{alg:module}).

\begin{example}
 Two modules $M_1 = \langle \mathcal{G}_{M_1}, \chi_{M_1} \rangle $ and $M_2 = \langle \mathcal{G}_{M_2}, \chi_{M_2} \rangle$ are \emph{identical} if $\mathcal{G}_{M_1}$ and $\mathcal{G}_{M_2}$ are isomorphic and the component types match, via, e.g., a graph isomorphism algorithm~\cite{west_introduction_2000}. 
We need to solve the sensor placement problem only once for any set of identical modules.
\end{example}

For example, a 3-bit ripple carry adder that uses one half-adder and two full-adders is modular (fig.~\ref{fig:3adder}). The two full adders are identical, i.e, have the same structure and behavior. %
The 3-bit adder gets as input two 3-bit Boolean numbers $a_2a_1a_0$ and $b_2b_1b_0$ and returns the sum $s_2s_1s_0$ and the carry $c_3$. The input and the outputs are observable (i.e., have sensors). However, the intermediate carry bits are not observable.
The adders consist of blocks of four NAND gates (fig.~\ref{fig:half} and fig.~\ref{fig:full}).  
The goal is to determine whether one of the blocks has failed in the 3-bit adder.
The output may be correct even though there is a fault in one of the small adders. 
To determine whether one of the blocks is faulty, we place sensors on components to obtain relevant observations.

For the larger $3$-bit adder, we use the sensor locations obtained by considering each module to be an independent system.
The positions of the sensors for the half and full adders are marked in yellow in fig.~\ref{fig:half} and fig.~\ref{fig:full}, respectively. 
Finally, during runtime, we use observations obtained from the sensors in conjunction with the model $\R_{\text{3-bit adder}}$ to find faulty components.

 The following theorem shows that sensors placed according to the modular sensor placement algorithm are sufficient to identify any faulty component at runtime.
\begin{theorem}
\label{thm:main}
Given a modular system consisting of modules $M_1,\cdots,M_q$ and $\textsf{COMP}$ the set of sensors returned by algorithm~\ref{alg:module}, and if $x$ is a faulty component, then $\textsf{faulty}(x)$ belongs to the stable model of solving the ASP encoding $\Pi_{\textsf{active}}$.
\end{theorem}
\begin{proof}
To prove the correctness of the modular approach, we introduce the notion of \emph{dependency graph}. The vertices of this graph correspond to modules of the system. There is an edge between two vertices of this dependency graph if a pair of components exist in the original graph with an edge between them.
The definition of modules (definition~\ref{defn:module}) implies that the dependency graph is \emph{acyclic}. A module is a \emph{root} if it does not have any in-edges in the dependency graph. 

We prove the correctness of algorithm~\ref{alg:module} by induction on the maximum distance of a module from any root.

\smallskip
\noindent \textit{Base case:} Any faulty component in a root module needs at most $k'$ configurations. 
The above assertion is a direct result of the correctness of the sensor placement algorithm (algorithm~\ref{alg:sensor}). More specifically, if any module is assumed to be an independent system, any faulty component in the module can be identified by algorithm~\ref{alg:active} with $k'$ configurations.

\smallskip
\noindent \textit{Inductive hypothesis:} Any faulty component in a module  that is $d$ distance from a root module can be diagnosed with at most $d\cdot k'$ configurations.

\smallskip
\noindent \textit{Induction step:} We need to show that any faulty component in a module that is at a distance of $d+1$ from a root can be diagnosed with at most $(d+1)k'$ configurations.
By the inductive hypothesis, the faulty components in all the modules at most $d$ distance away from the root can be determined using $d\cdot k'$ configurations. 
This implies that every other component's state/value (in these modules) can also be determined. 
Now, consider a module $M$ that is $d+1$ distance from the root; the inputs to this module are known. Therefore, this module $M$ can be treated as an independent system.
Any fault in this module can be identified with additional $k'$ configurations.
\end{proof}

\vspace{0.4em}

\begin{figure}[t!]
    \centering
    \includegraphics[width=0.8\linewidth]{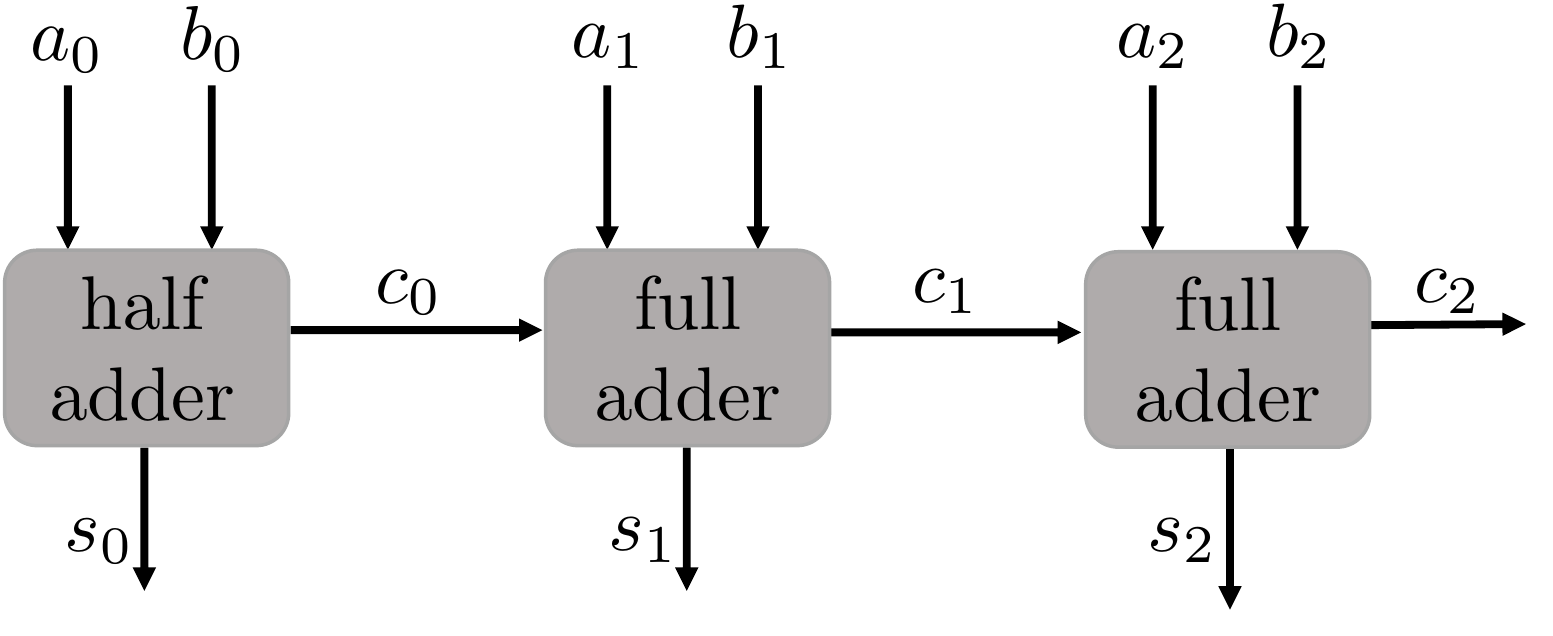}
    \caption{A 3-bit ripple carry adder consisting of two full-adders and one half-adder.}
    \label{fig:3adder}
\end{figure}

\begin{figure}[t!]
    \centering
    \includegraphics[width=0.55\linewidth]{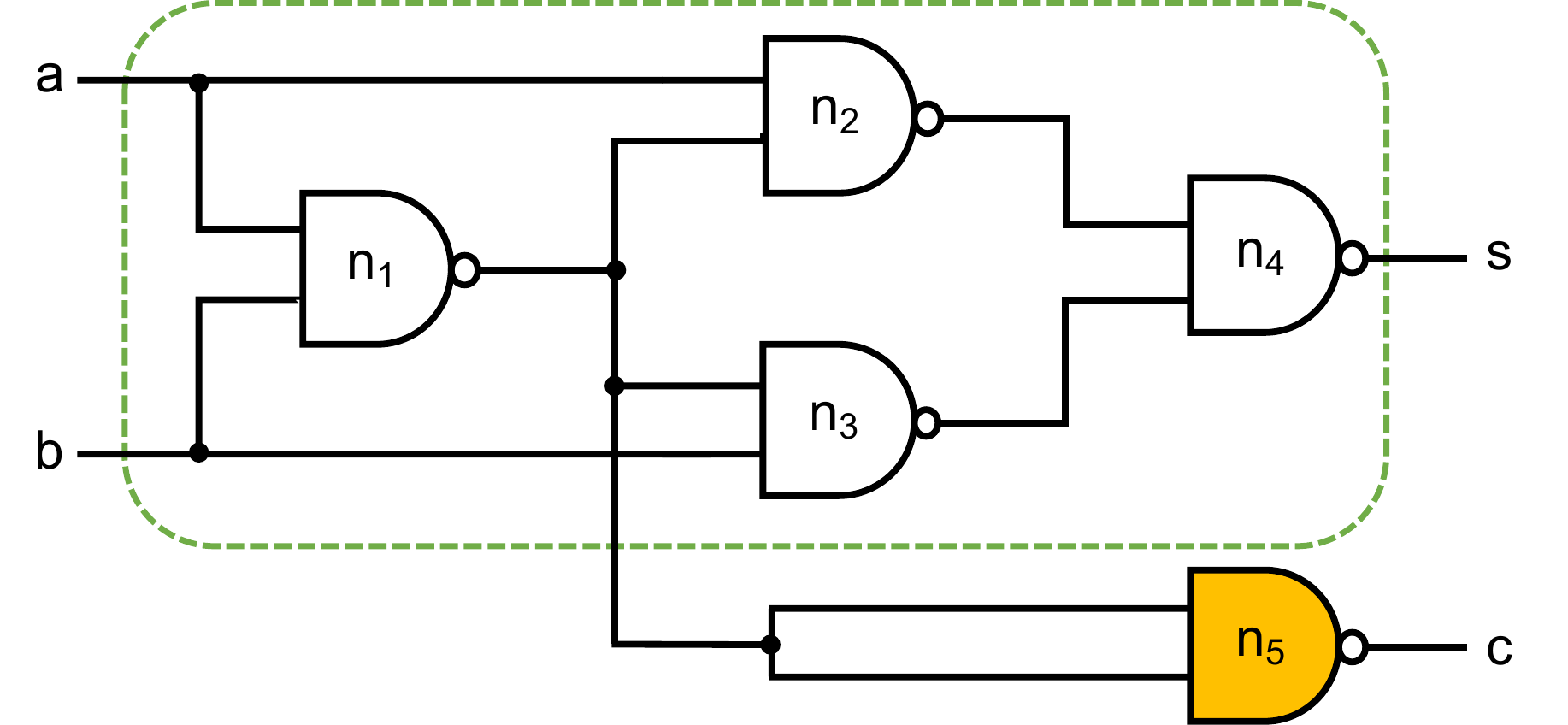}
    \caption{The half-adder requires a single sensor at $n_5$ to identify whether the fault is in one of the components in the green box.}
    \label{fig:half}
\end{figure}

\begin{figure}[t!]
    \centering
    \includegraphics[width=1\linewidth]{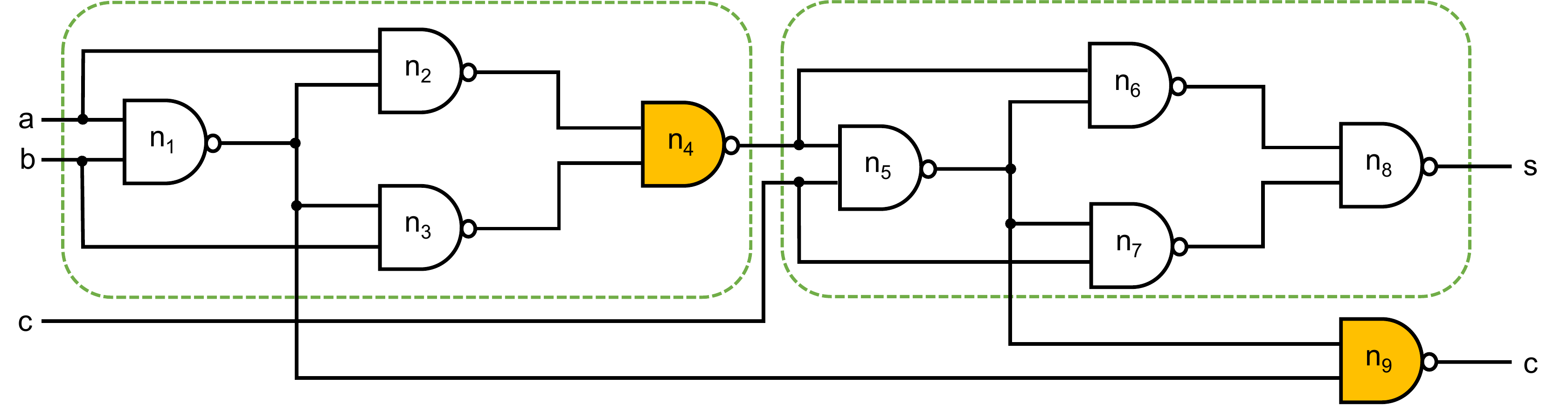}
    \caption{The locations of the sensors obtained by independently solving the sensor placement problem for the full-adder are marked in yellow. Two sensors are sufficient to identify whether some component in either of the two green boxes is faulty. }
    \label{fig:full}
\end{figure}

\section*{Case study: Electrical power system}
\label{sec:electric}
\label{sec:case_study}

We analyze the sensor placement problem for a system model of a single-line diagram of an electrical power system (adapted from Honeywell~\cite{Duces2007Oct}) comprising generators, batteries, rectifier units, and loads (fig.~\ref{fig:sensors} and fig.~\ref{fig:larger_ckt}). 

\begin{figure}[bh!]
    \centering
    \includegraphics[width=0.9\linewidth]{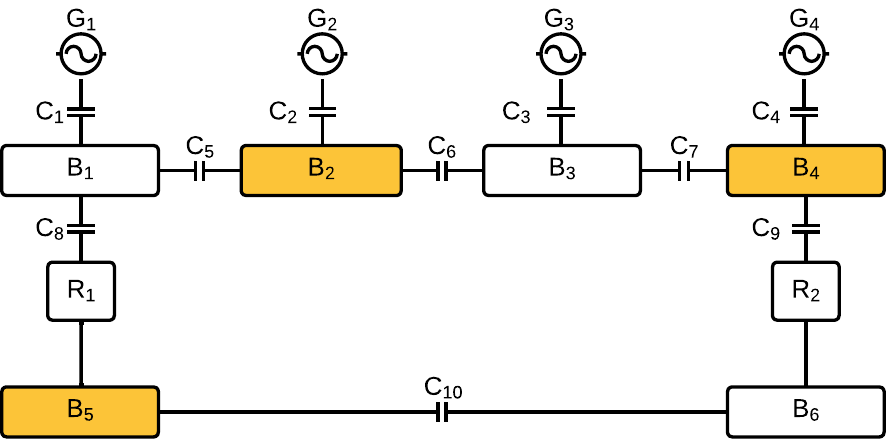}
    \caption{The components with sensors are highlighted in yellow.}
    \label{fig:sensors}
\end{figure}

\smallskip
\subsubsection*{Structure of the electrical power system}
The electrical power system in the figure has the following components ($V$) buses $B_1,\cdots,B_6$, generators $G_1,\cdots,G_4$, switches $C_1,\cdots,C_{10}$ and rectifiers $R_1, R_2$.
As discussed previously, the predicate $\textsf{component}$ is constructed by  
$\{ \textsf{component}(X)\mid X \in V \}.$
The connections are specified by the $\textsf{link}$ predicate given by the set 
$\{ \textsf{link}(X,Y)\mid (X,Y) \text{ is an edge in } E \}.$

\noindent Lastly, type of the component is defined by the set 
\begin{align*}
    &\{ \textsf{type}(X, \text{bus})\mid X \in \{B_1,\cdots,B_6\} \} \\ 
    \cup~& \{\textsf{type}(X, \text{generator})\mid X \in \{G_1,\cdots,G_4\} \} \\
    \cup~& \{\textsf{type}(X, \text{switch})\mid X \in \{C_1,\cdots,C_{10}\} \} \\
    \cup~& \{\textsf{type}(X, \text{rectifier})\mid X \in \{R_1,R_2)\}\}.
\end{align*}

\subsubsection*{Behavior of the electrical power system}
\label{sec:orig_electric}

We assume without loss of generality that there are no line faults. Hence, a switch has power if its incoming line has power.

From the perspective of diagnosis, there are some subtle differences between digital circuits and electrical power circuits. 
First, the links in electrical power systems can be altered by the use of switches. 
Second, there is no predefined notion of outputs. The outputs correspond to observations from the sensors.
Lastly, unlike a NAND gate, which has a fixed number of inputs and outputs, a load (bus) in the electrical power system may be connected to any number of components. Hence, the model (rules) for the electrical power system differs from the model for the half-adder.

Each component (type) in the electrical system has four possible states (values) that capture whether the component is healthy and powered. The $\textsf{value}$ function can be directly used for the system rules ($\R_\text{adder}$ in section~\ref{sec:basic}). Alternatively, we can also use an additional predicate $\textsf{powered}$ to simplify the rule from an expressibility standpoint~\cite{FerrarisL05}. The simplified formulation is
\begin{align*}
    \textsf{powered}(Y) \leftarrow \textsf{component}(Y),  \textsf{healthy}(Y),  \textsf{powered}(X), \\ \textsf{on}(X,Y), \textsf{switch}(X,Y), \textsf{link}(X,Y). 
\end{align*}

This rule simply states that a component is powered if it is connected to another component that is powered and the switch between them is turned on. 
The interpretation of the predicates in the above rule is presented in Table~\ref{tab:predicates}. 

If there is no switch in the link $(X,Y)$, we drop the predicate $switch$ in the previous rule and obtain the following rule.
\begin{align*}
    \textsf{powered}(Y) \leftarrow \textsf{component}(Y),  \textsf{healthy}(Y), \\ \textsf{powered}(X), \textsf{link}(X,Y). 
\end{align*}

We denote the above two rules for the electrical power system by $\R_\text{electric}$.
It is sufficient to qualify the fault status of component $Y$ and not $X$. Indeed, due to assumption~2 (section~\ref{sec:basic}), if component $X$ is faulty, then the predicate $\textsf{powered}(X)$ cannot be proved.

We analyze the power system (fig.~\ref{fig:sensors}) under the environment assumptions listed below. These assumptions are fairly common and have been adopted from the reference~\cite{XuTM15}. 
\begin{itemize}
    \item The buses $B_5$ and $B_6$ are essential and, therefore, are always healthy. This assumption is modeled by the following constraints.
$\bot \leftarrow~\textsf{faulty}(B_5)$, and $\bot \leftarrow \textsf{faulty}(B_6).$

\item One of the two rectifiers is  always healthy, i.e.,  
$\textsf{healthy}(R_1) \leftarrow \textsf{faulty}(R_2)$ and $\textsf{healthy}(R_2) \leftarrow \textsf{faulty}(R_1)$.

\item One of the left generators ($G_1$ or $G_2$) and one of the right generators ($G_3$ or $G_4$) are healthy, i.e., 
$\textsf{healthy}(G_1) \leftarrow \textsf{faulty}(G_2)$, $\textsf{healthy}(G_2) \leftarrow \textsf{faulty}(G_1)$, $\textsf{healthy}(G_3) \leftarrow \textsf{faulty}(G_4)$ and $\textsf{healthy}(G_4) \leftarrow \textsf{faulty}(G_3)$.
\end{itemize}

\begin{table}[t!]
    \centering
        \caption{The set of configurations for diagnosis.}
    \begin{tabular}{cl}
      \toprule
      \textbf{Configuration id} \qquad  & \textbf{Switches that are on} \\ \midrule
       1 & $C_2, C_4$ \\
       2 & $C_1, C_3, C_5, C_7$ \\
       3 & $C_1, C_8$ \\
       4 & $C_2, C_5, C_8$ \\
       5 & $C_3,C_5,C_6,C_8$ \\
       6 & $C_5,C_6,C_7,C_8$ \\
       7 & $C_4,C_9,C_{10}$ \\
       8 & $C_3,C_7,C_9,C_{10}$ \\
       9 & $C_2,C_6,C_7,C_9,C_{10}$ \\
       10 & $C_1, C_5,C_6,C_7,C_9,C_{10}$ \\
        \bottomrule
    \end{tabular}
    \label{tab:config}
\end{table}

The procedure (i.e., solving the ASP encoding~$\Pi_\text{active-sensor}$) returns predicates $\textsf{sensor}(b_2)$, $\textsf{sensor}(b_4)$,  and $\textsf{sensor}(b_5)$
 as a solution to sensor location (fig.~\ref{fig:sensors}) and the predicates
\begin{align*}
    & \textsf{on}(1, C_2), \textsf{on}(1, C_4), \textsf{on}(2, C_1), \textsf{on}(2, C_3), \textsf{on}(2, C_5), \\
    & \textsf{on}(2, C_7), \textsf{on}(3, C_1), \textsf{on}(3, C_8), \textsf{on}(4, C_2), \textsf{on}(4, C_5), \\
    & \textsf{on}(4, C_8), \textsf{on}(5, C_3), \textsf{on}(5, C_5), \textsf{on}(5, C_6), \textsf{on}(5, C_8), \\
    & \textsf{on}(6, C_5), \textsf{on}(6, C_6), \textsf{on}(6, C_7), \textsf{on}(6, C_8), \textsf{on}(7, C_4), \\
    & \textsf{on}(7, C_9), \textsf{on}(7, C_{10}), \textsf{on}(8, C_3), \textsf{on}(8, C_7), \textsf{on}(8, C_9), \\
    & \textsf{on}(8, C_{10}), \textsf{on}(9, C_2), \textsf{on}(9, C_6), \textsf{on}(9, C_7), \textsf{on}(9, C_9), \\
    & \textsf{on}(9, C_{10}), \textsf{on}({10},C_1), \textsf{on}({10}, C_5), \textsf{on}({10}, C_6),  \\
    & \textsf{on}({10}, C_7), \textsf{on}({10}, C_9), \text{ and } \textsf{on}({10}, C_{10}).
\end{align*}
for the  sequence of configurations (interpreted in table~\ref{tab:config}).%

The simplified electrical power system (fig.~\ref{fig:sensors}) represents a more extensive electrical power system (fig.~\ref{fig:larger_ckt}) of a Boeing 747~\cite{Duces2007Oct}. In addition, to the previously mentioned four types, the system has transformer-rectifier units (\textsf{TRU}), and alternating current transformers (\textsf{ACT}). The buses are characterized as \textsf{AC}/\textsf{DC} buses and have the characteristic of either high or low voltage. Some buses are essential, i.e., they always require power, marked as \textsf{ESS}. In this case study, these characterization model safety constraints (section~\ref{sec:safety}).

\begin{table*}[p]
    \caption{Collection of the rules for modeling active diagnosis and sensor placement.}
    \centering
    \resizebox{0.82\textwidth}{!}{
    \begin{tabular}{cl}
      \toprule
      \textbf{Number} \qquad  & \textbf{Rule and its interpretation} \\ \midrule
       1 & \textbf{guess a set of healthy components} \\
         & $\{\textsf{healthy}(X)\} \leftarrow \textsf{component}(X)$ \\ \midrule
       2 & \textbf{if a component is not healthy, then it is faulty} \\
         & $\textsf{faulty}(X) \leftarrow \textsf{component}(X), \textsf{not healthy}(X)$ \\ \midrule
       3 & \textbf{when is the output of nand 1, i.e., true?} \\
         & $\textsf{value}(\textsf{out}(X),1)\leftarrow \textsf{type}(X,\text{nand}), \textsf{healthy}(X), \textsf{value}(\textsf{in}_1(X),0)$ \\ \midrule
       4 & \textbf{the value of an faulty component} \\
         & $\textsf{value}(\textsf{out}(X),0)\leftarrow \textsf{component}(X), \textsf{faulty}(X)$ \\ \midrule
       5 & \textbf{there are no wire faults} \\
         & $\textsf{value}(Y,V) \leftarrow \textsf{link}(X,Y), \textsf{value}(X,V)$ \\ \midrule
       6 & \textbf{if the output of a nand gate is not true, then it is false} \\
         & $\textsf{value}(\textsf{out}(X),0)\leftarrow  \textsf{not value}(\textsf{out}(X),1)$ \\ \midrule
       7 & \textbf{find the minimal diagnosis} \\
         & $\#minimize\{1,X:\textsf{faulty}(X)\}$ \\ \midrule
       8 & \textbf{the expected value of each component is computed on the fly at runtime} \\
         & $\textsf{value}(I,Y, \text{on}) \leftarrow \textsf{value}(I,X, \text{on}), \textsf{link}(X,Y), \textsf{on}(I,X,Y), \textsf{switch}(X,Y)$ \\
         & $\textsf{value}(I,X, \text{on}) \leftarrow \textsf{source}(X),  \textsf{config}(I)$ \\ \midrule
       9 & \textbf{a component cannot be both healthy and faulty} \\
         & $\bot \leftarrow \textsf{healthy}(X), \textsf{faulty}(X)$ \\ \midrule
      10 & \textbf{the guess has to match the sensor information} \\
         & $\bot \leftarrow  \textsf{not healthy}(X), \textsf{input\_healthy}(I,X)$ \\
         & $\bot \leftarrow \textsf{not faulty}(X), \textsf{input\_faulty}(I,X)$ \\ \midrule
      11 & \textbf{the inferred value of the components} \\
         & $\textsf{inferred}(I,Y, \text{on}) \leftarrow \textsf{inferred}(I,X, \text{on}),\textsf{link}(X,Y),\textsf{on}(I,X,Y),\textsf{switch}(X,Y),\textsf{healthy}(Y)$ \\ \midrule
      12 & \textbf{inference rule for source} \\
         & $\textsf{inferred}(I,X, \text{on}) \leftarrow \textsf{healthy}(X),\textsf{source}(X)$ \\ \midrule
      13 & \textbf{if a component is faulty, it cannot transmit power across} \\
         & $\textsf{inferred}(I,X, \text{off})\leftarrow \textsf{faulty}(X)$ \\ \midrule
      14 & \textbf{construct fault scenarios} \\
         & $ \textsf{scenario}(S)\leftarrow 1 <= I <= n$ \\
         & $\{\textsf{fault}(S,X):\textsf{component}(X)\}\leftarrow  \textsf{scenario}(S)$  \\
         & $X!=Y\leftarrow \textsf{fault}(S1,X), \textsf{fault}(S2,Y), \textsf{not} S1 = S2$ \\ \midrule
      15 & \textbf{come up with k configurations for the switches} \\
         & $ \textsf{config}(I)\leftarrow1 <= I <= k$ \\
         & $\{\textsf{on}(I,X,Y):\textsf{switch}(X,Y)\}\leftarrow  \textsf{config}(I)$ \\ \midrule
      16 & \textbf{the expected value of each component w.r.t. scenario and configuration} \\
         & $\textsf{value}(S,I,Y, \text{on}) \leftarrow \textsf{value}(S,I,X, \text{on}), \textsf{link}(X,Y), \textsf{on}(I,X,Y), \textsf{switch}(X,Y),  \textsf{not fault}(S,Y),  \textsf{config}(I),  \textsf{scenario}(S)$ \\
         & $\textsf{value}(S,I,X, \text{on}) \leftarrow \textsf{source}(X),  \textsf{not fault}(S,X),  \textsf{scenario}(S), \textsf{config}(I)$ \\ \midrule
      17 & \textbf{input values depend on the sensors} \\
         & $\textsf{input\_healthy}(S,X)\leftarrow \textsf{sensor}(X),  \textsf{not fault}(S,X)$ \\
         & $\textsf{input\_faulty}(S,X)\leftarrow \textsf{sensor}(X), \textsf{fault}(S,X)$ \\ \midrule
      18 & \textbf{the sensors also return the value of the component} \\
         & $\textsf{inferred}(S,I,X, \text{on})\leftarrow \textsf{sensor}(X), \textsf{value}(S,I,X, \text{on})$ \\
         & $\textsf{inferred}(S,I,X, \text{off})\leftarrow \textsf{sensor}(X), \textsf{value}(S,I,X, \text{off})$ \\ \midrule
      19 & \textbf{guess a set of healthy components} \\
         & $\{\textsf{healthy}(S,I,X)\}\leftarrow \textsf{component}(X),  \textsf{config}(I),  \textsf{scenario}(S)$ \\ \midrule
      20 & \textbf{if a component is not healthy, then it is faulty} \\
         & $\textsf{faulty}(S,I,X)\leftarrow \textsf{component}(X),  \textsf{config}(I), \textsf{scenario}(S),  \textsf{not healthy}(S,I,X)$ \\ \midrule
      21 & \textbf{a component cannot be both healthy and faulty} \\
         & $\bot \leftarrow \textsf{healthy}(S,I,X), \textsf{faulty}(S,I,X)$ \\ \midrule
      22 & \textbf{the guess has to match the sensor information} \\
         & $\bot \leftarrow \textsf{not healthy}(S,I,X), \textsf{input\_healthy}(S,X),  \textsf{config}(I),  \textsf{scenario}(S)$ \\
         & $\bot \leftarrow \textsf{not faulty}(S,I,X), \textsf{input\_faulty}(S,X),  \textsf{config}(I),  \textsf{scenario}(S)$ \\ \midrule
      23 & \textbf{the inferred value of the components} \\
         & $\textsf{inferred}(S,I,Y, \text{on})\leftarrow \textsf{inferred}(S,I,X, \text{on}),\textsf{link}(X,Y),\textsf{on}(I,X,Y),\textsf{switch}(X,Y),\textsf{healthy}(S,I,Y),  \textsf{config}(I),  \textsf{scenario}(S)$ \\ \midrule
      24 & \textbf{inference rule for source} \\
         & $\textsf{inferred}(S,I,X, \text{on})\leftarrow \textsf{healthy}(S,I,X),\textsf{source}(X), \textsf{config}(I), \textsf{scenario}(S)$ \\ \midrule
      25 & \textbf{if a component is faulty, it cannot transmit power across} \\
         & $\textsf{inferred}(S,I,X, \text{off})\leftarrow \textsf{faulty}(S,I,X),  \textsf{config}(I), \textsf{scenario}(S)$ \\ \midrule
      26 & \textbf{a component is faulty, if it faulty in all the configurations} \\
         & $\textsf{faulty}(S,X)\leftarrow \textsf{component}(X),  \textsf{config}(I), \textsf{faulty}(S,I,X),  \textsf{not config}(I1), \textsf{healthy}(S,I1,X),  \textsf{scenario}(S)$ \\ \midrule
      27 & \textbf{consistency rule: a faulty component should not be diagnosed as healthy} \\
         & $\bot \leftarrow \textsf{fault}(S,X), \textsf{healthy}(S,X)$ \\
         \midrule
     28 & \textbf{place $\mathbf{m}$ sensors} \\
     & $\{\textsf{sensor}(X)\}=m \leftarrow \textsf{component}(X)$ \\
        \bottomrule
    \end{tabular}
    }
    \label{tab:rules}
\end{table*}

We consider the sensor placement problem for this four module system (marked in different colors in fig.~\ref{fig:larger_ckt}) using the \emph{modular} approach.
The structure of this electrical power system can be encoded in a similar manner. Further, the types of components in this larger system are the same as the smaller system. Therefore, rules describing the system and safety constraints remain unchanged. 
The rules (table~\ref{tab:rules}) required for solving the sensor placement problem are as in section~\ref{sec:sensor}. 

We solve the sensor placement problem just once for every module. None of the modules are identical. We first use the fixed the point method to find the modules of the system. These modules are highlighted in different colored boxes in the figure. We solved the sensor placement problem using algorithm~\ref{alg:module} (sensor positions marked in yellow in fig.~.\ref{fig:larger_ckt}).
The algorithm returned the positions of the sensors for the entire system in less than $10$ seconds. 

\begin{figure}[thb!]
    \centering
    \includegraphics[width=0.9\linewidth]{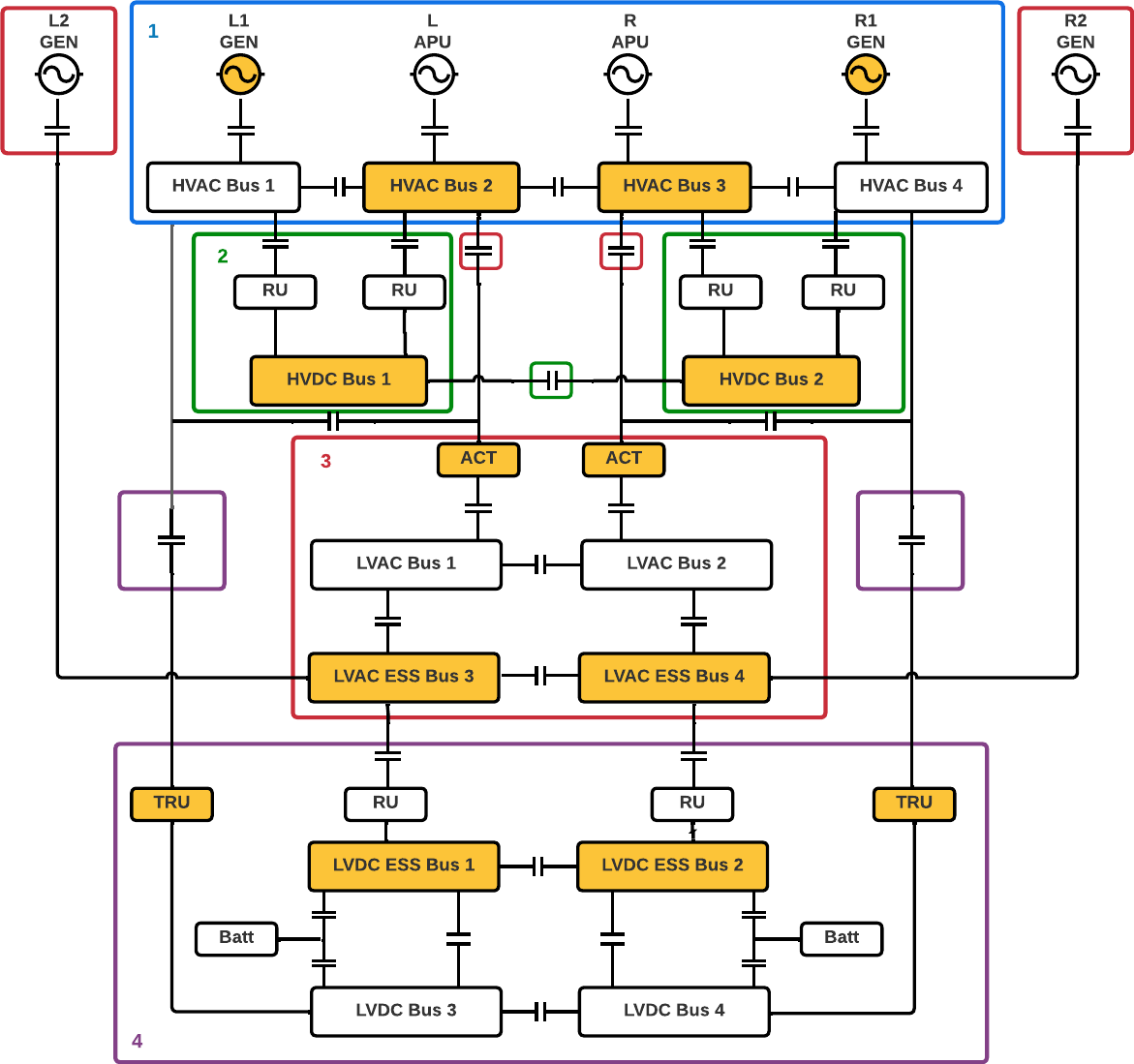}
    \caption{The four modules of the electric power system are highlighted in different colors. The sensors positions computed as a result of algorithm~\ref{alg:module} are highlighted in yellow.}
    \label{fig:larger_ckt}
\end{figure}

\section{Related work}
\label{sec:related}

Due to the high economic impact of faults, substantial work on the automated diagnosis and the placement of sensors varies in both their theoretical backdrop and the philosophy of design and implementation~\cite{Blanke}. 

The existing methods for sensor placement essentially find a minimal set of variables to be measured/monitored through the analysis of either a) the structure of the system or b) the redundancy of variables~\cite{Trave-Massuyes2006Oct,Rosich2007Dec}. 
Some works formulate the sensor placement problem as integer linear programs due to the  combinatorial nature of this problem. 
In Rosich et.al.~\cite{Rosich2009Jan} and Fijany et.al.~\cite{Fijany2006Mar}, fault diagnosis are defined by means of linear inequality constraints. Most of the work on fault diagnosis has study procedures for performing diagnosis (given a set of sensors) rather than determining the locations for placing sensors~\cite{Rosich2007Dec}. 
Furthermore, very little existing work uses logic solvers for sensor placement.

Reiter et.al. in their seminal paper demonstrate the use of logic models for diagnosis~\cite{reiter:87}.  
Subsequent work on model-based approaches has proved highly effective~\cite{Gao2015}.  
In this context, Prolog was used for diagnosis in~\cite{FriedrichGN90}. However, the propositional nature of Prolog severely restricts the type of systems that can be diagnosed. This limitation led to the use of constraint solvers in place of logic solvers~\cite{FelfernigSZ12, FelfernigWR15}.

Despite the ease of describing a system, consistency-based reasoning suffers from high computational complexity. However, the availability of cheap computing power together with modern SAT/ASP solvers has led to a resurgence of this technique for self-adaptive autonomous systems~\cite{Wotawa2019Mar, Wotawa2020Jun}. 
Recently, Wotawa et.al. focused on whether current ASP solvers can be used for diagnosis~\cite{Wotawa2022Apr}. 
Their reported runtime results indicate that ASP solvers are fast enough for diagnosis and also provide enough expressiveness of the input language for modeling. However, there is a lack of existing work on a consistency-based formulation for sensor placement.

When the sensors are already fixed in the system, it is possible to model electrical power systems through a discrete event system (DES) by encoding the observed values from the sensor into the states (of the DES) and using the switch configuration to define the transition between the states. 
Thus, the DES approach is enough to model the diagnosis problem when active network reconfiguration is permitted~\cite{Sampath1996, SampathLT98}. 
However, this approach is unsuitable for the sensor placement problem. 
For one, the size of a DES describing the system is exponential, and the combinatorial choice for the locations of the sensors makes this problem more complex. 
The proposed algorithm for the active diagnosis problem returns $k$ configurations. This bound $k$ is also used in \textit{anytime diagnosis algorithms}, i.e.,  algorithms with a time budget~\cite{Felfernig2018Aug}.

\section{Conclusion and future work}

Motivated by the usefulness of ASP in model-based diagnosis,
we extend the use of ASP for sensor placement. %
The main advantage of the proposed approach is the reusability of the system's rules. 
More specifically, the rules for fault scenarios, switch configurations, and consistency of diagnosis can be described independently. Further, we demonstrate the ease of modifying the rules of the system to account for faulty components and auxiliary sensors.
Our approach shifts the hardness from designing specific algorithms to formulating the system rules in first-order predicate logic (or in ASP). 
Describing the system rules is feasible since we do not need to describe  incorrect behaviors and fault models.%

Experiments show that our procedure is successful in finding sensor locations for systems with around 500 components in a few minutes. 
Additionally, the basic fault diagnosis algorithm using ASP is applicable for systems with $10^4$ components~\cite{Wotawa2022Apr}. 
Combining these two facts, if the systems are modular and the size of these modules is around 500, then the sensor placement problem for the entire system can be solved by independently solving the sensor placement problem for each module. This procedure only minimizes the number of sensors in each module (locally) which is greedy. %

We assume that sensors never fail.  However, there is a need to relax this assumption in order to better model real-world scenarios. 
While it is possible to  encode the sensors as separate components, this trick does not work when the sensors are yet to be placed.
An additional difficulty is the availability of runtime information about the faulty sensors. We need new assumptions to solve the problem when the sensors themselves have to be diagnosed.
As an extension of this work, we plan to address the \emph{robustness} of sensor placement. %

\printbibliography

\appendices

\end{document}